\documentclass[runningheads]{llncs}

\newif\ifarxivversion
\arxivversiontrue
\usepackage[a4paper,margin=1.1in]{geometry}
\usepackage[T1]{fontenc}
\usepackage{graphicx}
\usepackage{latexsym}
\usepackage{amsfonts}
\usepackage{amsmath,amssymb}
\usepackage{mathtools}
\usepackage{hyperref}
\usepackage[normalem]{ulem}
\usepackage[linesnumbered,lined,ruled,noend,vlined]{algorithm2e}
\usepackage{cleveref}
\usepackage{comment}
\usepackage{cite}
\usepackage{url}
\usepackage[full]{complexity}
\usepackage{lscape}
\usepackage{xspace}
\usepackage{tikz}
\usepackage{placeins}
\usepackage{multirow}
\usepackage{ifthen}
\usepackage{arydshln}
\usepackage{easy-todo}
\ifarxivversion
  \def\apxmode{inline} 
  \usepackage[appendix=\apxmode,bibliography=common,forwardlinking=no]{apxproof}
  \AddToHook{cmd/appendix/after}{}
\else
  \usepackage[appendix=strip,bibliography=common]{apxproof}
\fi
\newtheoremrep{theorem}{Theorem}
\newtheoremrep{lemma}[theorem]{Lemma}
\spnewtheorem{observation}[theorem]{Observation}{\bfseries}{\itshape}

\graphicspath{{images/}}

\newcommand{\subtreeiso}{\preceq}
\newcommand{\treeiso}{\simeq}
\newcommand{\set}[1]{\{#1\}}
\newcommand{\seq}[1]{(#1)}

\newcommand{\inset}[2]{\{#1 \mid #2\}}

\newcommand{\size}[1]{| #1 |}

\newcommand{\MCT}{\mathit{MCT}}
\newcommand{\mct}{\mathit{mct}}

\newcommand{\mis}{\mathit{MIS}}

\crefname{problem}{Problem}{Problems}
\crefname{observation}{Observation}{Observations}

\Crefname{algocf}{Algorithm}{Algorithms}

\makeatletter
\@ifpackageloaded{algorithm2e}{
    \SetKwProg{Fn}{Function}{}{}
    \SetKwProg{Procedure}{Procedure}{}{}
    \SetKwProg{Subprocedure}{Subprocedure}{}{}
    \SetKwComment{tcc}{//}{}
    \SetKw{Continue}{continue}
    \SetKwFunction{Output}{Output}%
    \SetKwFunction{EnumMC}{EMC}
    \SetKwInOut{AlgInput}{Input}
    \SetKwInOut{AlgOutput}{Output}
    \SetKwInOut{AlgPrecondition}{Pre-conditions}
    \SetKwInOut{AlgInvariant}{Invariants}
}{}
\makeatother

\newcommand{\ms}[1]{{\ifthenelse{\equal{#1}{}}{{ ms}}{{
ms}\xspace\left(#1\right)}}} %
\newcommand{\li}[1]{{\ifthenelse{\equal{#1}{}}{\ell}{\ell\xspace\left(#1\right)}}}
\newcommand{\posms}[1]{{\ifthenelse{\equal{#1}{}}{{ p}}{{
p}\xspace\left(#1\right)}}} %

\definecolor{darkred}{rgb}{0.75,0,0}

\ifarxivversion
  \newcommand{\beginmarkedresult}[1]{\begin{#1}}
  \newcommand{\prooflocation}{}
\else
  \newcommand{\beginmarkedresult}[1]{\begin{#1}[$\bigstar$]}
  \newcommand{\prooflocation}{The proofs marked with $\bigstar$ are available in the full version on arXiv.}
\fi

\makeatletter

\newcommand{\problemtitle}[1]{\gdef\@problemtitle{#1}}%
\newcommand{\probleminput}[1]{\gdef\@probleminput{#1}}%
\newcommand{\problemoutput}[1]{\gdef\@problemoutput{#1}}%
\newenvironment{problemblock}{%
    \problemtitle{}%
    \probleminput{}%
    \problemoutput{}%
}{%
    \par\addvspace{.5\baselineskip}
    \noindent
    \fbox{%
        \begin{minipage}{\dimexpr\linewidth-2\fboxsep-2\fboxrule\relax}
            \textsc{\@problemtitle}\par
            \noindent\textbf{Input:} \@probleminput\par
            \noindent\textbf{Output:} \@problemoutput
        \end{minipage}%
    }
    \par\addvspace{.5\baselineskip}
}

\makeatother

\let\doendproof\endproof
\renewcommand\endproof{~\hfill$\qed$\doendproof}

\begin{document}
\title{The Complexity of Maximal and Closed Frequent Tree Mining on Bounded Height Trees}
\titlerunning{The Complexity of Maximal and Closed Frequent Tree Mining}

\author{Kenta Komoto\inst{1}\orcidID{0009-0007-5580-8882} \and Kazuhiro Kurita\inst{2}\orcidID{0000-0002-7638-3322} \and Hirotaka Ono\inst{1}\orcidID{0000-0003-0845-3947}}
\authorrunning{K. Komoto et al.}
\institute{Nagoya University, Nagoya, Japan\\
\email{\{komoto.kenta.v1@s.mail.nagoya-u.ac.jp, ono@nagoya-u.jp\}}
\and
Okayama University, Okayama, Japan\\
\email{k-kurita@okayama-u.ac.jp}}

\maketitle

\begin{abstract}
  We address the problem of enumerating all 
  maximal and closed frequent trees.
  Although many practical algorithms have been developed for this problem,
  its complexity under ``realistic assumptions'' on tree height has
  not been clarified.
  More specifically, while it was known that the mining problem
  becomes hard when the tree height is at least $60$, the complexity for cases where the tree height is smaller has not
  yet been clarified.
  We resolve this gap by establishing both tractability and hardness results
  for ordered and unordered rooted trees of bounded height.
  In particular, some variants admit polynomial-delay algorithms,
  while others remain hard even when the tree height is only $2$ or $3$.

\end{abstract}

\section{Introduction}\label{sec:intro}
Pattern mining is a classical and central problem in the field of data mining.
The goal of this problem is to discover ``important'' or
``interesting'' patterns.
However, since it is difficult to mathematically formalize what
constitutes an ``important pattern'' or an
``interesting pattern,'' even the formulation of pattern mining
problems is nontrivial.
Therefore, many studies adopt high frequency as a necessary condition
for a pattern to be considered important or interesting.

Among the discrete structures considered in pattern mining, itemsets
are the most common.
The frequent itemset mining problem, which enumerates highly frequent
itemsets, is formulated as follows.
Given a set of elements $U = \set{1, \ldots, n}$,
a collection of $m$ subsets $X_1, \ldots, X_m \subseteq U$, and an
integer threshold $\theta$,
enumerate all subsets of $U$ that are contained in at least $\theta$ subsets.
For a subset $Y \subseteq U$, the number of itemsets $X_i$ satisfying $Y \subseteq X_i$ is called the frequency of $Y$.
If the frequency of $Y$ is at least $\theta$, then $Y$ is a frequent itemset.

In frequent itemset mining,
problems that enumerate more informative itemsets, such as closed
frequent itemsets and maximal frequent itemsets, have also been studied.
A frequent itemset $P$ is \emph{closed} if, for any itemset $P' \supset P$,
the frequency of $P'$ is smaller than that of $P$.
A frequent itemset $P$ is called \emph{maximal} if it is not strictly
contained in any other frequent itemset.
These problems have been studied from both theoretical and practical
perspectives~\cite{Boros:2003,DBLP:conf/fimi/UnoKA04}.

For the mining problems described above,
the number of outputs may be exponential in the input size $n$, and
$\mathrm{poly}(n)$-time enumeration in total is impossible.
Thus, output-sensitive analyses that take the number of outputs into
account are commonly used for mining problems~\cite{Johnson:1988}.
If an algorithm runs in time polynomial in the sum of the input size
and the number of outputs, then the algorithm is called an
\emph{output-polynomial time algorithm}.
In addition, there is another evaluation criterion known as \emph{delay}.
This measure is defined as the worst-case time between outputting one
solution and the next solution.

Strings and graphs are discrete structures considered in pattern
mining problems.
The frequent subsequence mining problem is to enumerate all
subsequences that appear in at least $\theta$ strings as subsequences
in a set of strings.
Similarly, the frequent subgraph mining problem is to enumerate all
graphs that appear in at least $\theta$ graphs in a set of graphs.
As in itemset mining, the enumeration of closed and maximal frequent
subsequences and subgraphs has been
studied~\cite{DBLP:conf/icdm/YanH02,Kimelfeld:2013,DBLP:journals/pacmmod/BuzzegaCKKP25}.

In previous studies,
three types of discrete structures were considered in frequent
closed and maximal pattern mining problems.
This leads to six different problem settings.
Currently, the complexities of these problems are known.
Except for the enumeration of closed frequent itemsets,
all problems are known to be unsolvable in output-polynomial time
unless $\P =
\NP$~\cite{DBLP:journals/pacmmod/BuzzegaCKKP25,Horvath:2013,Boros:2003,DBLP:conf/fimi/UnoKA04}.

In closed and maximal frequent subgraph mining,
even the problem of finding a single closed and maximal frequent subgraph
is \NP-complete, which can be shown easily by a reduction from the
maximum clique problem.
Using the arguments in~\cite{Horvath:2013},
it follows that closed and maximal frequent subgraph mining is \NP-hard
even to find a single solution for many graph classes.
Therefore, the authors in~\cite{Kimelfeld:2013} studied the
complexity of this problem when restricted to trees, one of the
simplest graph classes.
They showed that even when the input consists of two trees,
closed and maximal frequent subgraph mining does not admit an
output-polynomial time algorithm unless $\P = \NP$.
Their reduction requires trees of height at least $60$;
however, this bound arises from their specific gadget construction and is not believed to be tight.
In contrast, for general tree databases, we show hardness already at height $2$.

Thus, maximal and closed frequent subgraph mining remains hard even when
the input graphs are restricted to (rooted) trees.
In contrast, when considering practical applications of mining
problems on rooted trees, the height of the trees is usually much smaller.
One application of closed and maximal frequent tree mining is the
analysis of XML files~\cite{Termier:2002,doi:10.3233/FUN-2005-661-203}.
XML is a markup language used for data exchange between different
systems and can be represented as labeled rooted trees.
In the synthetic data generated in the experiments by Zaki~\cite{Zaki:2002},
the depth of rooted trees was set to around $10$, and in the
real-world CSLOGS dataset,
the average string length of the string representation of rooted
trees is $23.3$~\cite{Zaki:2002}.
Moreover, for rooted tree mining, many practical algorithms have been
studied for various settings,
such as whether trees are labeled or ordered, and whether the output
consists of closed frequent subtrees or maximal frequent
trees~\cite{10.5555/998688.1007091,6137215,10.1007/978-3-540-24775-3_9,Balcazar:2010,doi:10.3233/FUN-2005-661-203}.

For these realistic problem settings, the complexity of these mining
problems has remained open.
We investigate the complexity of these problems when the height of
each tree is small.
We summarize our contribution in \Cref{tab:complexity-result}.

\prooflocation







\begin{table}[t]
    \centering
    \caption{
        Summary of our results, shown in bold.
        PT = polynomial-time solvable;
        DelayP = polynomial-delay enumeration;
        Dual-hard = an output-polynomial-time algorithm for the problem would
        imply an output-polynomial-time algorithm for \textsc{Dualization};
        N-OP = no output-polynomial-time algorithm exists unless $\P = \NP$.
        For $h=1$, PT follows from the total order by root degree.
        $\dagger$ No explicit proof is given in~\cite{Kimelfeld:2013},
        but the proof for the unordered case can be adapted.
    }
    \label{tab:complexity-result}

    \begin{tabular}{c|c|c|c|c}
        \hline
        & \multicolumn{2}{c|}{Closed}
        & \multicolumn{2}{c}{Maximal} \\ \hline
        & Unordered & Ordered & Unordered & Ordered \\ \hline\hline

        $h = 1$
        & \textbf{PT}
        & \textbf{PT}
        & \textbf{PT}
        & \textbf{PT}
        \\ \hline

        $h = 2$
        & \textbf{DelayP}
        & \multirow{2}{*}{\textbf{Dual-hard}}
        &
        \begin{tabular}{c}
            \textbf{PT} if $\theta = |\mathcal{T}|$ \\
            \textbf{N-OP} in general
        \end{tabular}
        & \multirow{2}{*}{\textbf{N-OP}}
        \\ \cline{1-2}\cline{4-4}

        $3 \le h < 60$
        & \textbf{Dual-hard}
        &
        & \textbf{N-OP}
        &
        \\ \cline{1-3}\cdashline{4-5}

        $h \ge 60$
        & N-OP ~\cite{Kimelfeld:2013}
        & N-OP$^\dagger$
        & N-OP ~\cite{Kimelfeld:2013}
        & N-OP$^\dagger$
        \\ \hline
    \end{tabular}
\end{table}

\section{Preliminaries}\label{sec:prelim}
\subsection{Rooted trees and subtree isomorphism}
Let $G = (V, E)$ be a graph.
For a graph $G$, we denote the set of vertices and edges of $G$ as
$V(G)$ and $E(G)$, respectively.
A sequence of vertices $P = (v_1, \ldots, v_k)$ is a $v_1$-$v_k$
\emph{path} if for any  $1 \le i < j \le k$,
$v_i \neq v_j$ and $\set{v_i, v_{i + 1}} \in E$.
The \emph{length} of a path $P$ is defined by the number of vertices
in $P$ minus one.
A graph $G$ is \emph{connected} if $G$ has a $u$-$v$ path for any
pair of vertices $u, v \in V$.

A sequence of vertices $P = (v_1, \ldots, v_k)$ is a \emph{cycle} if
$v_1 = v_k$, and $(v_1, \ldots, v_{k-1})$ is a path of length
at least $2$.
A connected graph with no cycles is a \emph{tree}, denoted by $T$.
A \emph{rooted tree} is a pair $(T, r)$, where $T$ is a tree and $r$
is a vertex of $T$.
We simply write $T$ for $(T, r)$ and denote its root by $r(T)$.
The \emph{height} of a vertex $v$ in $T$ is the length of the $r$-$v$ path.
The \emph{height} of $T$ is the maximum height of its vertices.
For the $r$-$v$ path $P=(v_1=r,\ldots,v_k=v)$,
$v_{k-1}$ is the \emph{parent} of $v$, denoted by $\mathit{par}(v)$,
and $v_1,\ldots,v_k$ are \emph{ancestors} of $v$.
A vertex $u$ is a \emph{child} of $v$ if $\mathit{par}(u)=v$,
and a \emph{descendant} of $v$ if $v$ is an ancestor of $u$.
We denote the set of children of $v$ by $\mathit{ch}(v)$.
A non-root vertex that has no children is a \emph{leaf}.
For rooted trees $T_1 = (V_1, E_1)$ and $T_2 = (V_2, E_2)$, $T_2$ is
a \emph{subtree} of $T_1$ if $V_2 \subseteq V_1$, $E_2 \subseteq E_1$, and
$u$ is a child of $v$ in $T_1$ if and only if $u$ is a child of $v$ in $T_2$.

Let $T$ be a rooted tree and $v$ be a vertex in $T$.
$D(v)$ is the set of descendants of $v$.
A rooted tree $(T(v), v)$ is the \emph{subtree of $T$ rooted at $v$},
where $T(v)$ is a tree $(D(v), \inset{\set{u, w}\in E(T) }{u, w\in D(v)})$.
For two rooted trees $T_1$ and $T_2$,
$T_1$ is \emph{subtree isomorphic} to $T_2$
if $T_2$ has a subtree $T'_2$ such that $r(T'_2) = r(T_2)$
and there is a bijection $\varphi:
V(T'_2) \to V(T_1)$ satisfying
$u$ is the parent of $v$ in $T'_2$ if and only if $\varphi(u)$ is the
parent of $\varphi(v)$ in $T_1$.
We denote it as $T_1 \subtreeiso T_2$.
If $T_1 \subtreeiso T_2$, we say that $T_1$ is \emph{contained} in
$T_2$ as a subtree.
If $T_1 \subtreeiso T_2$ and $T_2 \subtreeiso T_1$, then we denote it
as $T_1 \treeiso T_2$, and $T_1$ is \emph{isomorphic} to $T_2$.
Moreover, if $T_1 \subtreeiso T_2$ and $T_1 \not\treeiso T_2$, we denote $T_1 \prec T_2$.
The bijection $\varphi$ is a \emph{subtree isomorphism mapping}.
Especially, $\varphi$ is a \emph{tree isomorphism mapping} if $T_1$
is isomorphic to $T_2$.

We define rooted \emph{ordered} trees.
Let $T$ be a rooted tree, and for each vertex $v \in V(T)$,
a binary relation $\le_v$ is a total order on $\mathit{ch}(v)$.
An ordered tree is the tuple $(T, \{ \le_v \}_{v \in V(T)})$.
When no confusion arises, we denote it by $T$.
For two rooted ordered trees $T_1$ and $T_2$,
we say that $T_1$ is \emph{subtree isomorphic} to $T_2$
if $T_2$ has a subtree $T'_2$ such that $r(T'_2) = r(T_2)$
and there exists a bijection
$\varphi: V(T'_2) \to V(T_1)$ such that
for all $u,v \in V(T'_2)$,
$u$ is the parent of $v$ in $T'_2$ if and only if
$\varphi(u)$ is the parent of $\varphi(v)$ in $T_1$, and
for all $u,v \in V(T'_2)$ with $\mathit{par}(u)=\mathit{par}(v)$,
$u \le_{\mathit{par}(u)} v$ if and only if
$\varphi(u) \le_{\mathit{par}(\varphi(u))} \varphi(v)$.
If $T_1 \subtreeiso T_2$ and $T_2 \subtreeiso T_1$, we write
$T_1 \treeiso T_2$ and say that $T_1$ is \emph{isomorphic} to $T_2$.
Moreover, if $T_1 \subtreeiso T_2$ and $T_1 \not\treeiso T_2$, we denote $T_1 \prec T_2$.
The bijection $\varphi$ is a \emph{subtree isomorphism mapping}.
If $T_1$ is isomorphic to $T_2$, $\varphi$ is a \emph{tree isomorphism mapping}.

\subsection{Frequent, closed, and maximal trees}
For a sequence of rooted unordered trees $\mathcal T = \seq{T_1, \ldots, T_n}$,
the \emph{support}  of $T$ is the number of trees in $\mathcal T$
that contain $T$ as a subtree.
More precisely, the support of $T$ is defined by $\size{\inset{R \in
\mathcal T}{ T \subtreeiso R}}$.
We denote by $\mathcal T(T)$ the set of trees in $\mathcal T$ that contain $T$ as a subtree.
A rooted tree $T$ is \emph{$\theta$-frequent} if the support of $T$
is at least $\theta$.
In particular, a $\size{\mathcal T}$-frequent tree is called a \emph{common
tree} of $\mathcal T$.
A $\theta$-frequent tree $T$ is \emph{maximal} if, for any
$\theta$-frequent tree $T'$, $T \not\prec  T'$.
When no confusion arises, we  denote a (maximal) $\theta$-frequent
tree as a (maximal) frequent tree.
A maximal $\size{\mathcal T}$-frequent tree is called a
\emph{maximal common tree} of $\mathcal T$.
We denote the set of all maximal common trees of $\mathcal T$ by $\MCT(\mathcal T)$.
If $\MCT(\mathcal T)$ consists of a unique tree,
we denote this tree by $\mct(\mathcal T)$.
A rooted tree $T$ is \emph{closed} with respect to $\mathcal T$ if,
for any rooted tree $T'$ satisfying $T \prec T'$,
the support of $T$ is larger than the support of $T'$.
All the definitions above extend analogously to rooted ordered trees.

\subsection{Enumeration problems and complexity measures}

In this paper, we discuss the following four enumeration problems.

\begin{problemblock}
    \problemtitle{Maximal Frequent Tree Mining}
    \probleminput{A nonempty multiset of rooted unordered/ordered trees
    $\mathcal{T}$ and an integer $\theta$ with $1 \le \theta \le |\mathcal{T}|$.}
    \problemoutput{All maximal $\theta$-frequent unordered/ordered
    trees of $\mathcal{T}$.}
\end{problemblock}

\begin{problemblock}
    \problemtitle{Closed Frequent Tree Mining}
    \probleminput{A nonempty multiset of rooted unordered/ordered trees
    $\mathcal{T}$ and an integer $\theta$ with $1 \le \theta \le |\mathcal{T}|$. }
    \problemoutput{All closed $\theta$-frequent unordered/ordered trees
    of $\mathcal{T}$. }
\end{problemblock}

For an instance $I$ of an enumeration problem, let $\mathrm{Sol}(I)$
denote the set of solutions to be output.  An enumeration algorithm is
said to run in \emph{output-polynomial time} if its total running time is
bounded by a polynomial in $|I|+|\mathrm{Sol}(I)|$.  The \emph{delay} of
an enumeration algorithm is the maximum time between two consecutive
outputs, including the time before the first output and the time after
the last output.  An algorithm has \emph{polynomial delay} if this delay
is bounded by a polynomial in the input size.

Let $H = (V,\mathcal{E})$ be a hypergraph. A set $T \subseteq V$ is a
\emph{transversal} of $H$ if $T \cap E \neq \emptyset$ for every
$E \in \mathcal{E}$, and \textsc{Dualization} asks to enumerate all minimal
transversals of $H$. Its best known algorithm runs in quasi-polynomial time~\cite{DBLP:journals/jal/FredmanK96}, and whether it admits an output-polynomial time algorithm has
been open for over three decades~\cite{DBLP:journals/dam/EiterMG08}. 
A set $I \subseteq V$ is an
\emph{independent set} of $H$ if $E \not\subseteq I$ for every
$E \in \mathcal{E}$, and \emph{maximal} if no independent set properly contains
it; let $\mis(H)$ be the set of maximal independent sets of $H$. Since
$I$ is independent if and only if $V \setminus I$ is a transversal,
complementation is a bijection between $\mathrm{mis}(H)$ and the minimal
transversals of $H$, and hence enumerating $\mathrm{mis}(H)$ is equivalent to
\textsc{Dualization} with respect to output-polynomial time
solvability.
We say that a problem is \emph{Dual-hard} if an
output-polynomial time algorithm for it yields one for \textsc{Dualization}.

\section{Maximal Common Tree Mining}\label{sec:common-subtree-mining}
We address \textsc{Maximal Common Tree Mining}.
It is a special case of \textsc{Maximal Frequent Tree Mining} and
\textsc{Closed Frequent Tree Mining} by setting $\theta = \size{\mathcal T}$.
We show that when each rooted unordered tree has height of at most $2$,
the maximal common tree is uniquely determined.
This observation is essential for our algorithm in \Cref{sec:frequent-subtree-mining}.
In addition, we show the dual-hardness of
\textsc{Closed Frequent Tree Mining} in the unordered case if each tree has height of at most $3$.
For the ordered case, we show that
\textsc{Closed Frequent Tree Mining} is also dual-hard, even if each tree has height of at most $2$.

\subsection{The Unordered Case}
\subsubsection{The Tractable Case: Height at most \texorpdfstring{$2$}{2}}\label{sssec:tractable-h2}

For a tree $T$ with a height at most $2$, let $\ell_T$ be the number of children of
$r(T)$ and $d(T)=(d^T_1,\ldots,d^T_{\ell_T})$ be the non-increasing degree sequence of children in $ch(r(T))$.
Thus, a tree of height at most $2$ is uniquely determined by such a non-increasing degree sequence.

\beginmarkedresult{theoremrep}\label{thm:uniqueness-of-common-trees-of-n-trees}
  Let $\mathcal T$ be a nonempty sequence of trees of height at most $2$.
  Then $\mathcal T$ has a unique maximal common tree.
\end{theoremrep}
\begin{toappendix}
\begin{proof}
  Let $\ell$ be the maximum number of children of the root among the trees in $\mathcal T$.
  We use $(-1)$-padding to extend every degree sequence to length $\ell$,
  where $-1$ represents the absence of a child.

  For two trees $Q$ and $R$ of height at most $2$, $Q\subtreeiso R$
  if and only if $d^Q_i \le d^R_i$
  for every $1 \le i \le  \ell$, where the degree sequences are
  understood with $(-1)$-padding to length $\ell$.
  We define a sequence $\pi=(p_1,\ldots,p_\ell)$ by
  $p_i=\min_{1\le j\le \size{\mathcal T}} d^{T_j}_i$
  for every $1 \le i \le \ell$.
  Since each $d(T_j)$ is non-increasing, $\pi$ is non-increasing.
  Let $\pi'$ be the sequence obtained from $\pi$ by removing all $-1$,
  and let $M$ be the tree of height at most $2$ whose
  degree sequence is $\pi'$.

  By construction, $p_i \le d^{T_j}_i$ for any integers $i$ and $j$.
  Hence $M \subtreeiso T_j$ for every $1 \le j \le \size{\mathcal T}$, and
  $M$ is a common tree of $\mathcal T$.
  We show that every common tree of $\mathcal T$ is contained in $M$.
  Let $X$ be a common tree of $\mathcal T$.
  Since $X\subtreeiso T_j$ for every $j$, we have
  $d^X_i \le d^{T_j}_i$ for every $i$, where $d(X)$ is also understood with
  $(-1)$-padding to length $\ell$.
  Therefore, $d^X_i \le \min_{1\le j\le \size{\mathcal T}} d^{T_j}_i = p_i$.
  Hence $X \subtreeiso M$.
  Thus $M$ contains every common tree of $\mathcal T$.
  Therefore, $M$ is the unique maximum common tree of $\mathcal T$.
\end{proof}
\end{toappendix}

\subsubsection{Dual-Hardness of height at most \texorpdfstring{$3$}{3}}\label{ssec:dual-hardness-of-height3}

Let $\mathcal H = (V,\mathcal E)$ be a hypergraph.
We may assume that every hyperedge is nonempty; otherwise, $\mathcal{H}$ has no independent set.
In what follows, we identify $V$ with $\set{1,\ldots, n}$
and $\size{\mathcal E}$ as $m$.
We construct a set of trees $\mathcal T(\mathcal H)= \set{S,T_1,\ldots,T_{m}}$
such that maximal common trees of $\mathcal T(\mathcal H)$ are in one-to-one
correspondence with maximal independent sets of $\mathcal H$.
The gadgets used in the reduction are illustrated in \Cref{fig:common:unordered:gadgets}.
For each $i \in V$, we first define a tree $P_i$ and $\bar P_i$, called \emph{vertex selection gadgets}.
The root of $P_i$ has $i$ children $p_1, \ldots, p_i$.
Each child $p_j$ has $2(n - i) + 2$ leaf children.
Then, for any distinct $i, j \in V$, the trees $P_i$ and $P_j$ are
incomparable under rooted subtree isomorphism.
More precisely, both $P_i \not\subtreeiso P_j$ and $P_j \not\subtreeiso P_i$ hold.
In addition, we define $\bar P_i$ by removing one leaf from a child $p_1$.
We also define a tree $D$, called the \emph{dummy gadget}.
The root of $D$ has $n$ children, and each child has exactly $2n$ leaf children.
Since the root of $P_i$ has $i \le n$ children and each of them has $2(n-i)+2 \le 2n$ leaf children, $P_i \subtreeiso D$ for every $i \in V$.
Consequently, $\bar{P}_i \subtreeiso D$ also holds for every $i \in V$.
Notice that the following observation holds.

\begin{figure}[t]
  \centering
  \includegraphics[page=1, width=1\linewidth]{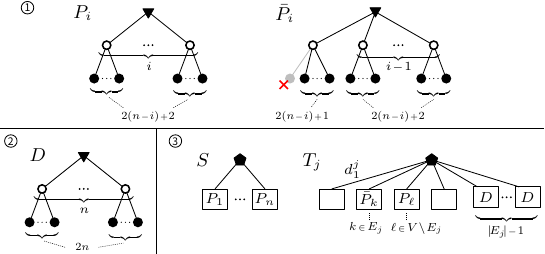}
  \caption{
    The vertex selection gadgets $P_i$ and $\bar{P}_i$ for $n$ and $\size{E_j}$.
    Vertices marked with the same symbol have the same height.
  }
  \label{fig:common:unordered:gadgets}
\end{figure}

\begin{observation}\label{obs:rigidity}
  For every $i\in V$,
  if a tree $T$ satisfies $T\subtreeiso P_i$ and
  $T\not\prec \bar P_i$, then $T\treeiso P_i$ or $T \treeiso \bar P_i$.
\end{observation}

We now define $\mathcal T(\mathcal H)$.
The root of $S$ has $n$ children $s_1, \ldots, s_n$. For every $i \in V$, $S(s_i) \treeiso P_i$.
For every hyperedge $E_j \in \mathcal{E}$, the root of $T_j$ has $n + \size{E_j} - 1$ children $d^j_1, \ldots, d^j_n$ and $\size{E_j} - 1$ dummy children $d^j_{n+1}, \ldots, d^j_{n + \size{E_j} - 1}$.
For every $i \in V$, we define $T_j(d^j_i) \treeiso \bar{P}_i$ if $i \in E_j$, and $T_j(d^j_i) \treeiso P_i$ otherwise. For every dummy child $d^j_k$, $T_j(d^j_k) \treeiso D$.
Clearly, every tree in $\mathcal T(\mathcal H)$ has height at most $3$.

For a set of vertices $U \subseteq V$, we define a tree $\mathcal F(U)$ as follows.
The root of $\mathcal F(U)$ has $n$ children $c_1,\ldots,c_n$.
For every $i \in V$, $\mathcal F(U)(c_i)\treeiso P_i$ if $i \in U$, and $\mathcal F(U)(c_i)\treeiso \bar P_i$ otherwise.

\begin{toappendix}
As an auxiliary lemma for proving \Cref{lem:MIS-to-MCT}, we establish the following lemma.
\begin{lemmarep}\label{lem:common-iff-independent}
  For a set of vertices $U \subseteq V$, $U$ is an independent set of
  $\mathcal H$ if and only if $\mathcal F(U)$ is a common tree of
  $\mathcal T(\mathcal H)$.
\end{lemmarep}
\begin{proof}
  ($\Rightarrow$) Suppose that $U$ is independent.
  Fix a hyperedge $E_j$.
  Since $U$ is independent, there exists a vertex
  $i \in E_j \setminus U$.
  Then,
  $\mathcal F(U)(c_i) \treeiso \bar P_i$,
  and hence $c_i$ can be mapped to the vertex child $d^j_i$ of $T_j$.
  For every $k \in E_j \setminus \set{i}$,
  we map $c_k$ to a distinct dummy child.
  This is possible because $T_j$ has $\size{E_j}-1$ dummy children and
  both $P_k$ and $\bar P_k$ are subtree isomorphic to $D$.
  For every $k \notin E_j$,
  we map $c_k$ to $d^j_k$.
  Therefore, $\mathcal F(U) \subtreeiso T_j$.
  Since this holds for every $j$ and $\mathcal{F}(U) \subtreeiso S$ by construction, $\mathcal{F}(U)$ is a common tree of $\mathcal{T}(\mathcal{H})$.

  ($\Leftarrow$) Suppose that $U$ is not independent.
  Then, there exists a hyperedge $E_j$ such that $E_j \subseteq U$.
  For every $v \in E_j$,
  $\mathcal F(U)(c_v) \treeiso P_v$.
  On the other hand,
  $T_j(d^j_v) \treeiso \bar P_v$, and
  $P_v \not\subtreeiso \bar P_v$.
  Moreover, $P_v$ cannot be mapped to any other vertex child,
  because $P_v$ and $P_k$ are incomparable for every $k \ne v$,
  and $\bar P_k \subtreeiso P_k$.
  Therefore, each of these $\size{E_j}$ children must be mapped to
  a dummy child.
  However, $T_j$ has only $\size{E_j}-1$ dummy children.
  Hence, $\mathcal F(U) \not\subtreeiso T_j$,
  and thus $\mathcal F(U)$ is not a common tree.
\end{proof}
\end{toappendix}

\beginmarkedresult{lemmarep}\label{lem:MIS-to-MCT}
  If $I$ is a maximal independent set of $\mathcal H$, then
  $\mathcal F(I)$ is a maximal common tree of $\mathcal T(\mathcal H)$.
\end{lemmarep}
\begin{toappendix}
\begin{proof}
  By \Cref{lem:common-iff-independent}, $\mathcal F(I)$ is a common tree.
  Suppose, for contradiction, that $\mathcal F(I)$ is not maximal.
  Then there exists a common tree $T'$ such that
  $\mathcal F(I) \prec T'$.
  Since $S$ has exactly $n$ children at its root and
  $\mathcal F(I)$ has $n$ children at its root,
  the tree $T'$ cannot be obtained by adding a new child to the root.
  Thus $T'$ must extend at least one child subtree of $\mathcal F(I)$.

  Since $\bar{P}_i \not\subtreeiso P_j$ for $i \neq j$ and
  $P_i$ is obtained from $\bar{P}_i$ by adding one leaf,
  there exists $i \notin I$ such that $\mathcal{F}(I \cup \set{i}) \subtreeiso T'$.
  Since $T'$ is a common tree, so is $\mathcal{F}(I \cup \set{i})$.
  By \Cref{lem:common-iff-independent}, $I \cup \set{i}$ is independent,
  contradicting the maximality of $I$.
  Hence $\mathcal F(I)$ is a maximal common tree.
\end{proof}
\end{toappendix}

We next prove the converse direction.
Let $T$ be a maximal common tree of $\mathcal T(\mathcal H)$.
Since $T$ is maximal and $T$ is contained in $S$,
$r(T)$ has $n$ children and there is a bijection between
$ch(r(S))$ and $ch(r(T))$.
In what follows, we assume that
children $ch(r(T))$ are indexed using
a subtree isomorphism mapping $\varphi$ from $V(T)$ to $V(S)$.
More precisely, the index of $c$ is $i$ if $\varphi(c) = s_i$, where
$\varphi$ is a subtree isomorphism mapping from $V(T)$ to $V(S)$.

\begin{toappendix}
As an auxiliary lemma for proving \Cref{lem:normal-form}, we establish the following lemma.

\begin{lemmarep}\label{lem:dummy-capacity}
  Let $T$ be a tree contained in $S$ such that
  $r(T)$ has exactly $n$ children.
  Index the children of $r(T)$ using a subtree isomorphism mapping
  from $T$ to $S$.
  Then $T$ is a common tree of $\mathcal T(\mathcal H)$ if and only if
  $\size{B_j(T)} \le \size{E_j}-1$ for every $j$, where $B_j(T) = \inset{i\in E_j}{T(c_i)\not\subtreeiso\bar P_i}$.
\end{lemmarep}
\begin{proof}
  Suppose that $\size{B_j(T)} \le \size{E_j} - 1$ for every $j$.
  Fix $E_j$.
  Map each $c_i$ with $i \in B_j(T)$ to a distinct dummy child of $T_j$,
  which is possible since $T(c_i) \subtreeiso P_i \subtreeiso D$.
  Map every remaining $c_i$ to the corresponding vertex child $d^j_i$.
  This is valid because
  $T(c_i) \subtreeiso \bar P_i$ for $i \in E_j \setminus B_j(T)$ and
  $T(c_i) \subtreeiso P_i$ for $i \notin E_j$.
  Hence $T \subtreeiso T_j$.
  Since this holds for every $j$, $T$ is a common tree.

  Conversely, suppose that $T$ is a common tree of
  $\mathcal T(\mathcal H)$.
  Fix a hyperedge $E_j$ and let
  $\varphi_j$ be a subtree isomorphism mapping from
  $T$ to $T_j$.
  Consider an arbitrary $i \in B_j(T)$.
  Since
  $T(c_i) \not\subtreeiso \bar P_i$
  and
  $T(c_i) \subtreeiso P_i$,
  \Cref{obs:rigidity} implies that
  $T(c_i) \treeiso P_i$.
  We show that $\varphi_j(c_i)$ must be a dummy child.
  Since $i \in B_j(T) \subseteq E_j$,
  the $i$-th vertex child of $T_j$ satisfies
  $T_j(d^j_i) \treeiso \bar P_i$.
  Hence, $c_i$ cannot be mapped to $d^j_i$ because
  $P_i \not\subtreeiso \bar P_i$.
  Moreover, $c_i$ cannot be mapped to any other vertex child
  $d^j_k$ with $k \ne i$.
  Indeed, $T_j(d^j_k)$ is isomorphic to either
  $P_k$ or $\bar P_k$.
  Since $\bar P_k \subtreeiso P_k$ and
  $P_i \not\subtreeiso P_k$ for $i \ne k$,
  we also have
  $P_i \not\subtreeiso \bar P_k$.
  Therefore, $\varphi_j(c_i)$ must be a dummy child.
  Thus, every $i \in B_j(T)$ is mapped to a distinct dummy child
  of $T_j$.
  Since $T_j$ has exactly $\size{E_j} - 1$ dummy children,
  we obtain
  $\size{B_j(T)} \le \size{E_j} - 1$.
  Since this holds for every $j$, the statement follows.
\end{proof}
\end{toappendix}

\beginmarkedresult{lemmarep}\label{lem:normal-form}
  Let $T$ be a maximal common tree of $\mathcal T(\mathcal H)$ such that
  $r(T)$ has exactly $n$ children.
  Then, $T(c_i)\treeiso P_i$ or $T(c_i) \treeiso \bar P_i$
  for every $i \in V$.
\end{lemmarep}
\begin{toappendix}
\begin{proof}
  By assumption, $r(T)$ has $n$ children.
  Since $T \subtreeiso S$, $T(c_i) \subtreeiso P_i$ for each $c_i$.
  If $T$ does not satisfy this lemma,
  there is a child $c_i$ such that $T(c_i) \prec \bar P_i$.
  Let $T^+$ be the tree obtained from $T$ by replacing $T(c_i)$  with $\bar P_i$.
  Then $T\prec T^+$ and $T^+ \preceq S$.

  We show that $T^+$ is a common tree.
  Let $T_j$ be a tree in $\mathcal T(\mathcal H)\setminus\set{S}$.
  If $i\in E_j$, then $i$ is not contained in $B_j(T^+)$ since
  $\bar P_i\preceq \bar P_i$.
  If $i\notin E_j$, then $i$ is not $B_j(T^+)$.
  Hence $B_j(T^+)=B_j(T)$.
  Since $T$ is common, \Cref{lem:dummy-capacity} gives
  $\size{B_j(T)}\le \size{E_j}-1$, and therefore
  $\size{B_j(T^+)}\le \size{E_j}-1$.
  By \Cref{lem:dummy-capacity}, $T^+\preceq T_j$.
  Since this holds for every $j$, $T^+$ is a common tree, contradicting the
  maximality of $T$.

  Therefore, for each $c_i$ in $T$,
  $T(c_i) \subtreeiso P_i$ and $T(c_i) \not \prec \bar P_i$.
  From \Cref{obs:rigidity},
  $T(c_i) \treeiso \bar P_i$ or
  $T(c_i) \treeiso P_i$.
\end{proof}
\end{toappendix}

For a maximal common tree $T$, we define
$\mathcal F^{-1}(T) = U$, where $U$ is the unique set satisfying $T \treeiso \mathcal F(U)$ from \Cref{lem:normal-form}.
Equivalently, $\mathcal F^{-1}(T) = \inset{i\in V}{T \text{ has a vertex } c_i \text{ such that } T(c_i) \treeiso P_i}$.

\beginmarkedresult{lemmarep}\label{lem:MCT-to-MIS}
  If $T$ is a maximal common tree of $\mathcal T(\mathcal H)$, then
  $\mathcal F^{-1}(T)$ is a maximal independent set of $\mathcal H$.
\end{lemmarep}

\begin{toappendix}
\begin{proof}
  Let $U = \mathcal F^{-1}(T)$.
  From \Cref{lem:normal-form}, $T \treeiso \mathcal F(U)$.
  Since $T$ is a common tree, $\mathcal F(U)$ is a common tree.
  From \Cref{lem:common-iff-independent}, $U$ is an independent set of $\mathcal H$.

  Suppose, for contradiction, that $U$ is not maximal.
  Then there exists $i \in V\setminus U$ such that $ U \cup\set{i}$ is independent.
  From \Cref{lem:common-iff-independent},
  $\mathcal F(U\cup\set{i})$ is a common tree.
  Moreover,
  $\mathcal F(U)\prec \mathcal F(U\cup\set{i})$,
  since $\mathcal F(U \cup \set{i})$ is obtained from $\mathcal F(U)$ by
  replacing the subtree $\bar P_i$ with $P_i$.
  This contradicts the maximality of $T\treeiso \mathcal F(U)$.
  Therefore, $U$ is maximal.
\end{proof}
\end{toappendix}

Combining \Cref{lem:MCT-to-MIS,lem:MIS-to-MCT},
it can be shown that $\mathcal F$ is a bijection between
$\MCT(\mathcal T(\mathcal H))$ and $\mis(\mathcal H)$,
and we obtain the following theorem.

\begin{theorem}
  If \textup{\textsc{Maximal Common Tree Mining}} for rooted unordered trees of height at most $3$ can be solved in output-polynomial time, then \textup{\textsc{Dualization}} can be solved in output-polynomial time.
\end{theorem}

\subsection{Dual-Hardness of the Ordered Case}

We next address the ordered case.
We also show the Dual-hardness of \textsc{Maximal Common Tree Mining} even if the height of each tree is at most $2$.
To this end,
we provide a reduction from a hypergraph $\mathcal H$ to a set of trees $\mathcal T(\mathcal H)$
such that there is an ``almost'' bijection between $\mis(\mathcal H)$ and $\MCT(\mathcal T(\mathcal H))$.
An almost bijection means that a bijection from
$\mis(\mathcal H)$ to $\MCT(\mathcal T(\mathcal H))$
can be constructed, except for a constant number of maximal common trees.

In our reduction, we assume that no vertex of $\mathcal{H} = (V, \mathcal{E})$ belongs to every hyperedge.
If such a vertex $v$ exists,
then $\mis(\mathcal{H}) = \set{ V \setminus \set{v} } \cup \set{ I \cup \set{v} \mid I \in \mis(\mathcal{H} - v) }$,
where $\mathcal H - v = (V \setminus \set{v}, \inset{E \setminus \set{v}}{E \in \mathcal E})$.
Thus, $v$ can be removed and the solutions reconstructed in polynomial time.
Hence, $\mathcal{H}$ has no independent set of size $n - 1$.
Furthermore, we identify each element of $V$ with an integer from $1$ to $n$.

We construct $m + 1$ trees $S, T_1, \ldots, T_{m}$ as follows.
The root $r(S)$ has $n$ children $(v_1, \ldots, v_{n})$.
Each child $v_i$ has one leaf as a child.
For each $E_i = \set{w_1, \ldots, w_{\size{E_i}}}$,
we define $T_i$ as follows.
Without loss of generality, we assume that $E_i$ is sorted.
The root $r(T_i)$ has $n + \size{E_i} - 1$ children $(v_1, \ldots, v_{n + \size{E_i} - 1})$.
The $k$-th child $v_k$ is a leaf if and only if $k = w_j + j - 1$ for some $1 \le j \le \size{E_i}$.
Otherwise, $v_k$ has a unique child $u_k$.
Hereafter, for $U \subseteq V$, we denote the tree obtained by the above construction as $T(U)$.
For $\mathcal H = (V, \mathcal E)$, we denote $\set{S, T_1, \ldots, T_m}$ as $\mathcal T(\mathcal H)$.

From the above construction, it is easy to show that each tree $T_i$ contains a tree $W$ such that
$r(W)$ has $n - 1$ children, with each child having a unique leaf.
Therefore, the following lemma holds.

\beginmarkedresult{lemmarep}\label{lem:mct:n-1}
  Let $W$ be a tree such that $r(W)$ has $n - 1$ children and each child has a unique leaf.
  If $\mathcal H$ has no independent sets with cardinality $n - 1$,
  then $W \in \MCT(\mathcal T(\mathcal H))$.
\end{lemmarep}
\begin{toappendix}
  \begin{proof}
    From the construction of each tree $T \in \mathcal T(\mathcal H)$,
    $W$ is a common tree of $\mathcal T(\mathcal H)$.
    Therefore, we only show the maximality of $W$.
    Suppose there is a tree $W'$ that contains $W$.
    Since $W'$ is contained in $S$, $r(W')$ has $n$ children.
    Moreover, if all children have a leaf, it contradicts that $W'$ is contained in $T_i \in \mathcal T(\mathcal H)$ for $1\le i \le m$
    since $r(T_i)$ has exactly $n-1$ vertices that have a child.
    Therefore, $r(W')$ has a leaf vertex $v$ as a child.

    Suppose that $v$ is the $j$-th child of $r(W')$.
    Since $W'$ is contained in $T_i \in \mathcal T(\mathcal H)$,
    there is an injection $\varphi: V(W')\to V(T_i)$.
    Since the root of $T_i$ has exactly $n-1$ vertices that have a child,
    for each $u \in ch(r(W')) \setminus \set{v}$, $\varphi(u)$ is uniquely determined.
    Therefore, if $j \notin E_i$, $W'$ is not subtree isomorphic to $T_i$.
    It implies that any $E_i \in \mathcal E$ contains $j$, and
    it contradicts the fact that $\mathcal H$ has no independent set with cardinality $n - 1$.
  \end{proof}
\end{toappendix}

From the above lemma, the root of every maximal common tree has $n$ children except for $W$.
Therefore, we can define the function $\mathcal{F}$ from the subtrees of $S$ such that the root has $n$ children to $2^V$.
For each subtree $T$ of $S$, $\mathcal{F}(T)$ contains $i$ if and only if the $i$-th child of $T$ is a non-leaf.
Here, $\mathcal{F}$ is a bijection, and we denote the tree defined by $U \subseteq V$ as $\mathcal{F}^{-1}(U)$.

\begin{toappendix}
As an auxiliary lemma for proving \Cref{lem:surjective,lem:tree-to-mis}, we establish the following lemma.
\begin{lemmarep}\label{lem:non-subtree-iso}
  For a set of vertices $U$ and $U'$, $U\subseteq U'$ if and only if $\mathcal{F}^{-1}(U') \not\subtreeiso T(U)$.
\end{lemmarep}
  \begin{proof}
    ($\Rightarrow$) We prove the contrapositive.
    Let $\varphi$ be a subtree isomorphism mapping from $V(\mathcal{F}^{-1}(U'))$ to $V(T(U))$.
    Since $r(T(U))$ has at most $n - 1$ non-leaf children,
    $\mathcal{F}^{-1}(U')$ has a leaf $v_i$ such that $\varphi(v_i)$ is a leaf.
    From the construction of $\mathcal{F}^{-1}(U')$, $v_i \not\in U'$.
    However, $U$ contains $v_i$ from the construction of $T(U)$.
    Therefore, $U \not\subseteq U'$.

    ($\Leftarrow$) We prove the contrapositive. Suppose that $U \not\subseteq U'$.
    Let the elements of $U$ be ordered as $(u_1, \ldots, u_k)$, and choose $u_q\in U\setminus U'$.
    We construct a subtree isomorphism mapping $\varphi$ from $\mathcal{F}^{-1}(U')$ to $T(U)$.
    By the construction of $T(U)$, the $(u_q+q-1)$-th child of the root of
    $T(U)$ is a leaf. Since $u_q\notin U'$, the $u_q$-th child of the root of
    $\mathcal{F}^{-1}(U')$ is also a leaf. We map the latter child to the former
    one.
    We map all remaining children of the root of $\mathcal{F}^{-1}(U')$, in
    order, to the non-leaf children of the root of $T(U)$. This is well-defined:
    the root of $T(U)$ has $|V|+k-1$ children, exactly $k$ of which are leaves,
    and hence it has $|V|-1$ non-leaf children. Moreover, this mapping preserves
    the order. Indeed, before the $(u_q+q-1)$-th child of the root of $T(U)$,
    there are exactly $u_q-1$ non-leaf children, and after it there are exactly
    $|V|-u_q$ non-leaf children.
    Finally, for each non-leaf child $w$ of $r(\mathcal{F}^{-1}(U'))$,
    we map its unique leaf child to the unique leaf child of $\varphi(w)$.
    Therefore, we obtain a subtree isomorphism from $\mathcal{F}^{-1}(U')$ to
    $T(U)$.
  \end{proof}
\end{toappendix}

\beginmarkedresult{lemmarep}\label{lem:surjective}
  Let $I$ be a maximal independent set of $\mathcal H$.
  Then, $\mathcal{F}^{-1}(I)$ is a maximal common tree of $\mathcal T(\mathcal H)$.
\end{lemmarep}
\begin{toappendix}
  \begin{proof}
    We first show that $\mathcal{F}^{-1}(I)$ is a common tree of $\mathcal T(\mathcal H)$.
    Suppose that $\mathcal{F}^{-1}(I)$ is not a common tree.
    It implies that $\mathcal{F}^{-1}(I) \not\subtreeiso T_i$ for some $E_i \in \mathcal E$.
    From \Cref{lem:non-subtree-iso}, $I$ is not an independent set since $I \supseteq E_i$.

    We next show that $\mathcal{F}^{-1}(I)$ is maximal.
    Suppose that $\mathcal{F}^{-1}(I)$ is not maximal.
    Let $T'$ be a maximal common tree containing $\mathcal{F}^{-1}(I)$.
    From the definition of $\mathcal{F}$,
    $\mathcal{F}(T') \supset I$.
    Since $I$ is a maximal independent set of $\mathcal H$, $\mathcal H$ has a hyperedge $E$ such that $E \subseteq \mathcal{F}(T')$.
    From \Cref{lem:non-subtree-iso},
    $\mathcal F^{-1}(\mathcal{F}(T')) \not\subtreeiso T(E)$, and
    it contradicts the fact that $T'$ is a common tree of $\mathcal T(\mathcal H)$.
  \end{proof}
\end{toappendix}

\beginmarkedresult{lemmarep}\label{lem:tree-to-mis}
  Let $T \not\treeiso W$ be a maximal common tree of $\mathcal T(\mathcal H)$.
  Then, $\mathcal{F}(T)$ is a maximal independent set of $\mathcal H$.
\end{lemmarep}
\begin{toappendix}
  \begin{proof}
    Suppose that $\mathcal{F}(T)$ is not an independent set, that is,
    $\mathcal H$ has $E \subseteq \mathcal{F}(T)$.
    From \Cref{lem:non-subtree-iso}, $\mathcal{F}^{-1}(E) \not\subtreeiso T(E)$,
    it contradicts the fact that $T$ is a common tree of $\mathcal T(\mathcal H)$
    since $\mathcal{F}^{-1}(E)$ is contained in $T$.

    Suppose that $\mathcal{F}(T)$ is not maximal.
    Let $I'$ be a maximal independent set of $\mathcal H$ containing $\mathcal F(T)$.
    From \Cref{lem:surjective}, $\mathcal{F}^{-1}(I')$ is a common tree of $\mathcal T(\mathcal H)$.
    From the definition of $\mathcal{F}$, $T \subtreeiso \mathcal{F}^{-1}(I')$.
    It contradicts the maximality of $T$.
  \end{proof}
\end{toappendix}

From \Cref{lem:tree-to-mis,lem:surjective},
$\mathcal{F}$ is a bijection between $\MCT(\mathcal{T}(\mathcal{H})) \setminus \{W\}$ and $\mis(\mathcal{H})$.
Moreover, for a rooted tree $T \not\treeiso W$,
we can compute $\mathcal{F}(T)$ in polynomial time.
Since $W$ is the only additional maximal common tree,
an output-polynomial time algorithm for \textsc{Maximal Common Tree Mining}
yields an output-polynomial time algorithm for enumerating $\mis(\mathcal{H})$.
Therefore, the following theorem holds.

\begin{theorem}
  If \textup{\textsc{Maximal Common Tree Mining}} for the rooted ordered trees of height at most $2$ can be solved in output-polynomial time,
  then \textup{\textsc{Dualization}} can be solved in output-polynomial time.
\end{theorem}

\section{Closed Frequent Tree Mining}\label{sec:frequent-subtree-mining}
We present a polynomial-delay and polynomial-space algorithm for
\textsc{Closed Frequent Tree Mining} in the unordered case
if each tree in $\mathcal T$ has height of at most $2$.

Before describing the details of our algorithm,
we show that \textsc{Closed Frequent Tree Mining} can be regarded as
the problem of enumerating a subset of $\mathcal T$.
In what follows, for $\mathcal T$ and a tree $T$, we denote the set of trees in $\mathcal T$ that contain $T$ as $\mathcal T(T)$.

\beginmarkedresult{lemmarep}\label{thm:bipartite}
    Let $\mathcal T$ be  a set of rooted trees with height of at most $2$.
    Then, a tree $T$ is a closed tree of $\mathcal T$ if and only if
    $T \treeiso T_{\mct}$, where $T_{\mct} \treeiso
    \mct(\mathcal T(T))$.
\end{lemmarep}
\begin{toappendix}
\begin{proof}
    For a closed tree $T$,
    if $T \prec  T_\mct$, then it contradicts the closedness of $T$.
    Thus, we show the other direction.

    Suppose that $T \treeiso T_\mct$.
    From \Cref{thm:uniqueness-of-common-trees-of-n-trees},
    $T_\mct$ is the tree whose associated integers are the
    pointwise minimum of $T'$ over all $T' \in \mathcal T(T)$.
    Suppose that $T$ is non-closed.
    In this case, there is a rooted tree $T'$ such that $T'$ contains
    $T$ and $\mathcal T(T) = \mathcal T(T')$.
    However, since $T \treeiso T_{\mct}$,
    $\mathcal T(T')$ is a proper subset of $\mathcal T(T)$. This
    contradicts the assumption that $T$ is non-closed, and hence $T$
    must be closed.
\end{proof}
\end{toappendix}

From \Cref{thm:bipartite}, \textsc{Closed Frequent Tree Mining} can be regarded as
the problem of enumerating all subsets of $\mathcal T' \subseteq
\mathcal T$ such that $|\mathcal{T}'| \ge \theta$ and $\mathcal T' = \mathcal T(\mct(\mathcal T'))$.

In what follows, we present an algorithm for enumerating all such
subsets based on reverse search~\cite{Avis:1996}.
We provide an overview of reverse search.
In reverse search, we define a tree structure $\mathcal F$ rooted at $R$.
The set of vertices in $\mathcal F$ is the set of closed frequent
trees, and we assume that we can find $R$ in polynomial time.
We enumerate all closed frequent trees by traversing $\mathcal F$ from $R$.
We define the parent of a closed frequent tree $T$ for each non-root
$T \in V(\mathcal F) \setminus \set{R}$.
To traverse $\mathcal F$,
the definitions of the neighbors of $T$ are essential.
As a key point in the definition of the neighbors,
the neighbors of $T$ must be defined so that they contain all closed
frequent trees that have $T$ as their parent.
When we define the parent and the neighbors satisfying these conditions,
we can traverse $\mathcal F$ from the root.
More precisely, if the parent and the neighbors can be found in polynomial time,
then we can enumerate all closed frequent trees in polynomial delay.

Our algorithm does not enumerate all closed trees but only the closed
frequent trees for a threshold $\theta$.
Therefore, when defining the parent–child relation, we ensure that
the support of a parent is at least that of its child.
This guarantees the monotonicity of the tree structure, which in turn
makes it possible to enumerate all closed frequent trees.
Several papers have proposed techniques that use monotonicity along
parent–child
relationships~\cite{DBLP:journals/algorithmica/Uno10,DBLP:journals/iandc/KobayashiKW25,DBLP:conf/wg/KobayashiKW22}.

We define the root, the parent, and
the neighbors of each solution.
As the root of $\mathcal F$, we define $R = \mct(\mathcal T)$.
From \Cref{thm:uniqueness-of-common-trees-of-n-trees}, it can be
found in polynomial time.
For a closed frequent tree $T \in V(\mathcal F) \setminus \set{R}$,
we define the parent of $T$ as follows.
Since $T$ is not the root $R$,
$\mathcal T \setminus \mathcal T(T) \neq \emptyset$.
We can find $\mathcal T(T)$ in polynomial time since subtree
isomorphism can be solved in polynomial time~\cite{MATULA197891}.
We choose $T'  \in \mathcal T \setminus \mathcal T(T)$ such that
$\mct(\mathcal T(T) \cup \set{T'})$ is maximal with respect
to a relation $\subtreeiso$.
If several such trees $T'$ exist,
we select one according to a predetermined rule.
The parent of $T$ is defined by $\mct(\mathcal T(T) \cup \set{T'})$.
In this definition, the support of the parent of $T$ is greater than
the support of $T$.
Thus, if $T$ is a closed frequent tree, then its parent is also a
closed frequent tree.
Moreover, by applying the parent relation at most $\size{\mathcal{T}}$ times,
any closed frequent tree reaches $R$.
This means that this parent relation does not make a cycle on $\mathcal F$.

We define the neighbors of each closed frequent tree $T$.
For each vertex $v \in V(T)$, let $T_v^+$ be the rooted tree obtained
from $T$ by adding a new leaf as a child of $v$.
If $T_v^+$ is $\theta$-frequent, we define
$\mct(\mathcal{T}(T_v^+))$ as a neighbor of $T$.
By \Cref{thm:bipartite}, $\mct(\mathcal{T}(T_v^+))$ is a closed tree,
and it is $\theta$-frequent since $\size{\mathcal{T}(T_v^+)} \ge \theta$.

We show that all children of $T$ are contained in neighbors of $T$.
To this end, we show the following lemma.

\beginmarkedresult{lemmarep}\label{lem:adj}
    Let $T$ be a closed tree and $T_c$ be a child of $T$.
    Then, there is a vertex $v$ such that $\mct(\mathcal
    T(T_v^+)) \treeiso T_c$.
\end{lemmarep}
\begin{toappendix}
\begin{proof}
    Since $T$ is the parent of $T_c$, it is maximal
    with respect to $\subtreeiso$ among the trees
    $\operatorname{mct}(\mathcal{T}(T_c) \cup \set{R})$
    for $R \in \mathcal{T} \setminus \mathcal{T}(T_c)$.
    Moreover, $T$ has a vertex $v$ such that $T_v^+ \subtreeiso T_c$
    since $T$ is subtree isomorphic to $T_c$.
    Since $T_v^+ \subtreeiso T_c$, we have
    $\mathcal T(T_c) \subseteq \mathcal T(T_v^+)$.
    Suppose that the inclusion is strict, and choose
    $R \in \mathcal T(T_v^+) \setminus \mathcal T(T_c)$.
    Then $T \prec T_v^+ \subtreeiso \operatorname{mct}(\mathcal T(T_c) \cup \{R\})$,
    contradicting the maximality in the definition of the parent of $T_c$.
    Hence $\mathcal T(T_v^+) = \mathcal T(T_c)$.
    Since $T_c$ is closed, we obtain
    $\operatorname{mct}(\mathcal T(T_v^+)) \simeq T_c$.
\end{proof}
\end{toappendix}

From \Cref{lem:adj}, we obtain a polynomial-delay and
polynomial-space algorithm using reverse search.

\begin{theorem}\label{thm:complexity-of-enumerate-closed-tree}
    If each tree of $\mathcal T$ has height of at most $2$,
    then \textup{\textsc{Closed Frequent Tree Mining}} can be solved in
    polynomial delay and polynomial space.
\end{theorem}

\section{Maximal Frequent Tree Mining}\label{sec:hardness-of-maximal-subtree-mining}

We show that \textsc{Maximal Frequent Tree Mining} has no
output-polynomial time algorithms
unless $\P = \NP$ in both cases.
To prove this, we introduce a technique to show
the hardness of the enumeration problems.
To show the hardness of the enumeration problem,
we consider the problem called the \emph{another solution problem}
(also called finished decision
    problem~\cite{DBLP:phd/basesearch/Bokler18}, or additional
problems~\cite{BABIN2017316}).
This discussion is used as folklore to show the hardness of enumeration
problems~\cite{DBLP:journals/siamcomp/LawlerLK80,DBLP:journals/ipl/BrosseDKLUW24,Boros:2003,DBLP:phd/basesearch/Bokler18,DBLP:journals/pacmmod/BuzzegaCKKP25,DBLP:conf/sand/Kurita0SU25}.
For details on the hardness of the another solution problem and
enumeration problems, see~\cite{DBLP:conf/sand/Kurita0SU25,DBLP:phd/basesearch/Bokler18,DBLP:journals/ipl/BrosseDKLUW24}.

\subsection{Hardness of the Unordered Case}

We show that there are no output-polynomial time algorithms for enumerating all maximal frequent trees of rooted unordered trees $\mathcal{T}$ even if the height of each tree is at most $2$.
To this end, we show the \NP-hardness of \textsc{Another Maximal Frequent Tree}.

We provide a reduction from \textsc{Another Maximal Frequent Itemset}.
An instance of this problem is a tuple
$(U,\mathcal X,\mathcal Y,\eta)$, where
$U = \set{1, \ldots, n}$ is a set of elements,
$\mathcal X = \set{X_1, \ldots, X_k}$ is a collection of itemsets,
$\mathcal Y = \set{Y_1, \ldots, Y_\ell}$ is a set of maximal frequent
itemsets, and $\eta$ is an integer threshold.
The task is to determine whether there exists a maximal $\eta$-frequent
itemset of $\mathcal X$ that is not contained in any itemset of
$\mathcal Y$.
Boros et al.~\cite{Boros:2003} showed that this problem is \NP-complete.
Our reduction encodes the membership of each item $i \in U$ by choosing one of the two consecutive degree $2i$ and $2i + 1$.

For a set $X \subseteq U$, we define its characteristic function
$\chi_X : U \to \set{0, 1}$ by setting $\chi_X(i) = 1$ if $i \in X$,
and $\chi_X(i) = 0$ otherwise.

For each $X \subseteq U$, we construct an unordered tree $\mathsf{T}(X)$ of height at most $2$ as follows.
The root of $\mathsf{T}(X)$ has exactly $n$ children, and for each $i \in U$, it has one child of degree $2i + \chi_X(i)$.
Equivalently, the degree sequence of the children of the root is
$d(\mathsf{T}(X)) = (2n + \chi_X(n), 2(n-1) + \chi_X(n-1), \dots, 2 + \chi_X(1))$.

Equivalently, for each $i \in U$, the subtree rooted at the corresponding
child is a star on $2i + 1$ vertices if $i \in X$, and a star on $2i$
vertices otherwise.
Using this encoding, we transform an instance $(U, \mathcal{X}, \mathcal{Y}, \eta)$ of \textsc{Another Maximal Frequent Itemset} into an instance $(\mathcal{T}, \mathcal{S}, \theta)$ of \textsc{Another Maximal Frequent Tree}, where
$\mathcal{T} \coloneqq \set{ \mathsf{T}(X) \mid X \in \mathcal{X} }$,
$\mathcal{S} \coloneqq \set{ \mathsf{T}(Y) \mid Y \in \mathcal{Y} }$,
and $\theta \coloneqq \eta$.
Each tree $\mathsf{T}(X)$ has $O(n^2)$ vertices, and hence the above
transformation can be computed in polynomial time.

We first show that the encoding preserves itemset containment.

\beginmarkedresult{lemmarep}\label{lem:maximal-unordered-containment}
    For any $X, X' \subseteq U$,
        $X \subseteq X'$ if and only if $\mathsf{T}(X) \subtreeiso \mathsf{T}(X')$.
\end{lemmarep}
\begin{toappendix}
\begin{proof}
    By the characterization of subtree isomorphism for rooted unordered trees of height at most $2$,
    $\mathsf{T}(X) \subtreeiso \mathsf{T}(X')$ holds if and only if every entry of the degree sequence of $\mathsf{T}(X)$ is at most the corresponding entry of the degree sequence of $\mathsf{T}(X')$.
    Since the $i$-th item is encoded by the degree $2i + \chi_X(i)$,
    this is equivalent to the following inequality for every $i \in U$:
        $2i + \chi_X(i) \le 2i + \chi_{X'}(i)$ if and only if 
        $\chi_X(i) \le \chi_{X'}(i)$.
    This holds if and only if $X \subseteq X'$.
\end{proof}
\end{toappendix}

As an immediate consequence of \Cref{lem:maximal-unordered-containment},
we obtain \Cref{obs:maximal-unordered-support}.

\begin{observation}\label{obs:maximal-unordered-support}
    For every $X \subseteq U$, the frequency of $X$ in $\mathcal{X}$ is equal to the support of $\mathsf{T}(X)$ in $\mathcal{T}$.
\end{observation}

We next show that every maximal frequent tree of the constructed instance
is represented by some itemset.
For a tree $T$, let $\mathcal{X}(T)$ denote the set of itemsets
$X \in \mathcal{X}$ such that $T \subtreeiso \mathsf{T}(X)$.

\beginmarkedresult{lemmarep}\label{lem:maximal-unordered-mfs-property}
    Let $T$ be a maximal $\theta$-frequent tree of $\mathcal{T}$.
    Then, there exists $X \subseteq U$ such that $T \treeiso \mathsf{T}(X)$.
\end{lemmarep}
\begin{toappendix}
\begin{proof}
    Since $T$ is $\theta$-frequent, $\size{\mathcal{T}(T)} \ge \theta$.
    By \Cref{thm:uniqueness-of-common-trees-of-n-trees},
    $\mathcal{T}(T)$ has a unique maximal common tree $M$.
    Since, $T$ is also common tree of $\mathcal{T}(T)$,
    we have $T \subtreeiso M$.
    Moreover, $M$ is contained in every tree in $\mathcal{T}(T)$.
    Thus, $M$ is also a $\theta$-frequent tree.
    Since $T$ is a maximal $\theta$-frequent tree, $T$ cannot be a proper subtree of $M$.
    Therefore, $T \treeiso M$.

    It remains to show that $M$ is represented by some itemset.
    We define
    \begin{align*}\textstyle
        X^* \coloneqq \bigcap_{X \in \mathcal{X}(T)} X.
    \end{align*}
    Since
    $\mathcal{T}(T) = \set{ \mathsf{T}(X) \mid X \in \mathcal{X}(T) }$,
    by the construction of the maximal common tree in
    \Cref{thm:uniqueness-of-common-trees-of-n-trees},
    the degree sequence of $M$ is obtained by taking the componentwise minimum of the degree sequences of the trees $\mathsf{T}(X)$ for $X \in \mathcal{X}(T)$.
    Therefore, for every $i \in U$, the entry corresponding to $i$ is
    \begin{align*}
        \min_{X \in \mathcal{X}(T)} \left( 2i + \chi_X(i) \right) = 2i + \chi_{X^*}(i).
    \end{align*}
    Here, the equality follows from the definition of $X^*$.
    Hence, $d(M) = d(\mathsf{T}(X^*))$, and thus $M \treeiso \mathsf{T}(X^*)$.
    Since $T \treeiso M$, we obtain $T \treeiso \mathsf{T}(X^*)$.
\end{proof}
\end{toappendix}

We are now ready to prove the correctness of the reduction.

\beginmarkedresult{lemmarep}\label{lem:maximal-unordered-reduction}
    There exists a maximal $\eta$-frequent itemset $X$ of $\mathcal{X}$ such that $X \not\subseteq Y$ for every $Y \in \mathcal{Y}$ if and only if there exists a maximal $\theta$-frequent tree $T$ of $\mathcal{T}$ such that $T \not\subtreeiso S$ for every $S \in \mathcal{S}$.
\end{lemmarep}
\begin{toappendix}
\begin{proof}
    ($\Rightarrow$) Let $X$ be a maximal $\eta$-frequent itemset of $\mathcal{X}$ such that $X \not\subseteq Y$ for any $Y \in \mathcal{Y}$.
    By \Cref{obs:maximal-unordered-support}, $\mathsf{T}(X)$ is $\theta$-frequent.

    We claim that $\mathsf{T}(X)$ is maximal.
    Otherwise, there exists a maximal $\theta$-frequent tree $T'$ such that $\mathsf{T}(X) \prec T'$.
    By \Cref{lem:maximal-unordered-mfs-property}, $T'$ is isomorphic to $\mathsf{T}(X')$ for some $X' \subseteq U$.
    Then, \Cref{lem:maximal-unordered-containment} implies $X \subsetneq X'$.
    Since $T'$ is $\theta$-frequent,
    \Cref{obs:maximal-unordered-support} further implies that $X'$ is
    $\eta$-frequent in $\mathcal{X}$,
    contradicting the maximality of $X$.

    It remains to show that $\mathsf{T}(X)$ is not contained in any tree of $\mathcal{S}$.
    Suppose that $\mathsf{T}(X) \subtreeiso S$ for some $S \in \mathcal{S}$.
    By the definition of $\mathcal{S}$, we have $S = \mathsf{T}(Y)$ for some $Y \in \mathcal{Y}$.
    Then, \Cref{lem:maximal-unordered-containment} gives $X \subseteq Y$, contradicting the choice of $X$.

    ($\Leftarrow$) Let $T$ be a maximal $\theta$-frequent tree of $\mathcal{T}$ such that $T \not\subtreeiso S$ for any $S \in \mathcal{S}$.
    By \Cref{lem:maximal-unordered-mfs-property}, there exists $X \subseteq U$ such that $T \treeiso \mathsf{T}(X)$.
    By \Cref{obs:maximal-unordered-support}, $X$ is $\eta$-frequent in $\mathcal{X}$.

    We claim that $X$ is maximal.
    Otherwise,
    there exists an $\eta$-frequent itemset $X'$ such that
    $X \subsetneq X'$.
    By \Cref{lem:maximal-unordered-containment},
    $\mathsf{T}(X) \prec \mathsf{T}(X')$
    and by \Cref{obs:maximal-unordered-support},
    $\mathsf{T}(X')$ is $\theta$-frequent.
    This contradicts the maximality of $T$.

    Finally, suppose that $X \subseteq Y$ for some $Y \in \mathcal{Y}$.
    By \Cref{lem:maximal-unordered-containment}, $T \treeiso \mathsf{T}(X) \subtreeiso \mathsf{T}(Y)$.
    Since $\mathsf{T}(Y) \in \mathcal{S}$, this contradicts the assumption that $T$ is not contained in any tree of $\mathcal{S}$.
\end{proof}
\end{toappendix}

By \Cref{lem:maximal-unordered-reduction}, the above construction gives
a polynomial-time reduction from \textsc{Another Maximal Frequent Itemset}
to \textsc{Another Maximal Frequent Tree} for rooted unordered trees of
height at most $2$.
Therefore, \textsc{Another Maximal Frequent Tree} is \NP-hard even under
this height restriction.
As discussed at the beginning of this section, an output-polynomial time
algorithm for \textsc{Maximal Frequent Tree Mining} would imply
a polynomial-time algorithm for the corresponding another-solution problem.
Therefore, we obtain the following theorem.

\begin{theorem}\label{thm:hardness-maximal-unordered}
    There are no output-polynomial time algorithms for \textup{\textsc{Maximal Frequent Tree Mining}} unless $\P = \NP$, even if the height of each rooted unordered tree is at most $2$.
\end{theorem}

\subsection{Hardness of the Ordered Case}

We show that there are no output-polynomial time algorithms for
enumerating all maximal frequent trees of rooted ordered trees
$\mathcal T$ even if the height of each tree is at most $2$.
To this end, we show the \NP-hardness of \textsc{Another Maximal
Frequent Tree} for ordered trees.

We provide a reduction from \textsc{Another Maximal Frequent Itemset},
which is the same source problem as in the unordered case.
Let $(U,\mathcal X,\mathcal Y,\eta)$ be an instance of this problem.
Our reduction is as follows.
We set the threshold $\theta$ to $\eta$.
The gadgets used in the reduction are illustrated in \Cref{fig:maximal:ordered:gadgets}.
For each itemset $X \subseteq U$, we generate a rooted ordered tree
$T(X)$ as follows.
The root of $T(X)$ has $n$ children.
If $X$ contains an element $j$, we add a leaf to the $j$-th child
of the root of $T(X)$.
As shorthand notation, we denote the union of the multiset of
resultant trees of a set of itemsets $\mathcal X$ as $\mathcal
T(\mathcal X)$.
We define a tree $R$ whose root has $n - 1$ children, each having a unique leaf child.
Let $\mathcal{R}$ be the multiset consisting of $\theta$ copies of $R$.
We denote the multiset of trees $\mathcal T(\mathcal{X}) \cup \mathcal
R$ as $\mathcal T_{\mathcal R}(\mathcal{X})$.
We set $\mathcal{S} = \mathcal{T}(\mathcal{Y})$ and add $R$ to $\mathcal{S}$ if $R$ is maximal frequent.

\begin{figure}[t]
    \centering
    \includegraphics[page=2, width=0.6\linewidth]{gadgets.pdf}
    \caption{
        The itemset gadget $T(X)$ and the tree $R$.
    }
    \label{fig:maximal:ordered:gadgets}
\end{figure}

From the above procedure, we obtain the instance of \textsc{Another Maximal Frequent Tree}.
In this construction, we obtain \Cref{obs:subtree_iso}.

\beginmarkedresult{lemmarep}\label{obs:subtree_iso}
    For two itemsets $X_1$ and $X_2$,
    $T(X_1)\subtreeiso T(X_2)$ if and only if $X_1 \subseteq X_2$.
\end{lemmarep}
\begin{toappendix}
    \begin{proof}
        Suppose that $X_1 \subseteq X_2$.
        We consider the following function $\varphi$ from $V(T(X_1))$ to $V(T(X_2))$.
        For the root $r(T(X_1))$, $\varphi(r(T(X_1))) = r(T(X_2))$.
        For the $i$-th child $c^i_1$ of $r_1$, $\varphi(c^1_i) = c^2_i$,
        where $c^2_j$ is the $j$-th child of $r(T(X_2))$.
        From the construction of $T(X_1)$,
        the remaining vertices consist of a single leaf that is a
        grandchild of the root.
        We denote a leaf $\ell^1_i$ with the parent $c^1_i$ as $v_i$.
        Since $X_1 \subseteq X_2$, $T(X_2)$ also contains a leaf
        $\ell^2_i$ with the parent $c^2_i$.
        We define $\varphi(\ell^1_i) = \ell^2_i$.
        Therefore, $T(X_1) \subtreeiso T(X_2)$.

        Suppose that $T(X_1)$ is subtree isomorphic to $T(X_2)$.
        Therefore, there is a subtree isomorphism mapping $\varphi$ from
        $V(T(X_1))$ to $V(T(X_2))$.
        Since $\varphi$ is a subtree isomorphism mapping,
        $\varphi(c^1_i) = c^2_i$.
        If $X_1 \setminus X_2 \neq \emptyset$,
        $X_1\setminus X_2$ has an item $j$.
        In this case, $c^1_j$ has a child $\ell^1_j$.
        On the other hand, $c^2_j$ is a leaf.
        Therefore, the parent of $\varphi(\ell^1_j)$ is not $c^2_j$, and
        it contradicts that $\varphi$ is a subtree isomorphism mapping.
    \end{proof}
\end{toappendix}

In the remainder of this subsection, we show that
$\mathcal X$ has a maximal frequent itemset not contained in
$\mathcal Y$ if and only if
$\mathcal T_{\mathcal R}(\mathcal X)$ has a maximal frequent tree not
contained in $\mathcal{S}$.

\beginmarkedresult{lemmarep}\label{lem:ordred:itemset-tree}
    Let $X$ be a maximal frequent itemset not contained in $\mathcal Y$.
    Then, $\mathcal T_{\mathcal R}(\mathcal X)$ has a maximal $\theta$-frequent tree $T$, and $T \not\subtreeiso S$
    for any $S \in \mathcal{S}$.
\end{lemmarep}
\begin{toappendix}
    \begin{proof}
        From \Cref{obs:subtree_iso},
        $T(X)$ is contained in a tree in $T(X_i)$ if and only if
        $X$ is contained in $X_i$.
        Since $X$ is $\eta$-frequent, $T(X)$ is $\theta$-frequent.
        Moreover, any tree in $\mathcal T(\mathcal Y)$ does not contain
        $T(X)$ as a subtree from \Cref{obs:subtree_iso}.
        Otherwise, it contradicts the maximality of $X$.
        Therefore, $T(X)$ is a maximal frequent tree that is not
        contained in $\mathcal T(\mathcal Y)$.
    \end{proof}
\end{toappendix}

\beginmarkedresult{lemmarep}\label{lem:ordered:tree-itemset}
    Let $T$ be a maximal frequent tree of $\mathcal T_{\mathcal R}(\mathcal X)$ not contained in $\mathcal{S}$.
    Then, there is an itemset $X$ with a frequency of at least $\theta$ and
    $X$ is not contained in any itemset in $\mathcal Y$.
\end{lemmarep}
\begin{toappendix}
\begin{proof}
    We show that $r(T)$ has exactly $n$ children.
    Suppose, for contradiction, that $r(T)$ has fewer than $n$ children.
    Then $T$ is contained in $R$.
    Since $R$ is $\theta$-frequent, the maximality of $T$ implies $T \treeiso R$.
    By construction, $R$ is then contained in $\mathcal{S}$, a contradiction.
    Therefore, $r(T)$ has exactly $n$ children.

    For the $i$-th child $c_i$ of $r(T)$, $c_i$ has at most one child.
    For $T$, we consider an itemset $X$ defined as follows.
    The $i$-th child $c_i$ has a child if and only if $X$ contains $i$.
    In this construction, $T(X) = T$.
    From \Cref{obs:subtree_iso}, $X$ is an itemset with a frequency of
    at least $\theta$ and $X$ is not contained in any itemset in $\mathcal Y$.
\end{proof}
\end{toappendix}

From \Cref{lem:ordered:tree-itemset,lem:ordred:itemset-tree},
if \textsc{Maximal Frequent Tree Mining} for the ordered case can be
solved in output-polynomial time,
\textsc{Another Maximal Frequent Itemset} can be solved in polynomial time.
Therefore, we obtain \Cref{thm:ordred:hardness}.

\begin{theorem}\label{thm:ordred:hardness}
    There is no output-polynomial time algorithm for \textup{\textsc{Maximal Frequent Tree Mining}}, even if each rooted ordered tree has height of at
    most $2$, unless $\P = \NP$.
\end{theorem}

\section{Conclusion}\label{sec:conclusion}

We investigate the complexity of frequent tree mining problems.
In the ordered case,
these problems are hard even if the height of each tree is at most $2$.
More precisely, enumerating the maximal frequent trees has no
output-polynomial time unless $\P = \NP$ even if the height of each
tree is at most $2$.
In addition, enumerating the maximal common trees is at least as hard as
\textsc{Dualization}.
In the unordered case, the enumeration of maximal frequent trees has no
output-polynomial time unless $\P = \NP$ even if the height of each
tree is at most $2$.
The enumeration of closed frequent trees is at least as hard as \textsc{Dualization}, even if
the height of each tree is at most $3$.
Moreover, the enumeration of closed frequent trees can be solved in
polynomial delay if the height of each tree is at most $2$.

\section*{Acknowledgements}
This work was supported by JSPS KAKENHI Grant Number 
JP22H00513, 
JP24K02898, 
JP25K21273, 
JP25K03077, 
JP25K03080, 
JP25K00136, and 
JP26K23808, 
and by CRONOS Grant Number JPMJCS24K2. 

\bibliographystyle{splncs04}
\bibliography{main}

\end{document}